\documentclass{article}
\usepackage{
	amsfonts,
	amsmath,
	amssymb,
	amsthm,
	array,
	chngpage,
	color,
	complexity,
	enumerate,
	float,
	graphicx,
	latexsym,
	mathrsfs,
	multicol,
	multirow,
	sectsty,
	subcaption,
	tikz,
	tikz-cd,
	verbatim,
	wrapfig,
	wasysym
}
\usepackage[english]{babel}
\usepackage{hyperref}

\usetikzlibrary{arrows,automata, positioning}
\usetikzlibrary{decorations.pathmorphing}
\tikzset{snake it/.style={decorate, decoration=snake}}
\allsectionsfont{\centering}

\newcommand{\eps}{\varepsilon}
\newcommand{\abs}[1]{\lvert#1\rvert}

\newcommand{\F}{\mathcal F}

\newcommand{\computer}{proved by computer}
\newtheorem{thm}{Theorem}

\newtheorem{con}[thm]{Conjecture}
\newtheorem{df}[thm]{Definition}

\theoremstyle{plain}

\newcommand{\pathc}{\pi}
\newcommand{\wordc}{\omega}

\DeclareMathOperator{\Acc}{Acc}
\DeclareMathOperator{\logAcc}{logAcc}

\newcommand{\ceil}[1]{\left\lceil#1\right\rceil}

\newcommand{\myUrl}[1]{
	\begin{center}
		{\small\url{#1}} 
	\end{center}
}

\begin{document}
	\title{Few paths, fewer words: model selection with
	automatic structure functions}
	\author{Bj{\o}rn Kjos-Hanssen}
	\maketitle
	\begin{abstract}
		We consider the problem of finding an optimal statistical model for a given binary string.
		Following Kolmogorov, we use structure functions.
		In order to get concrete results, we replace Turing machines by finite automata
		and Kolmogorov complexity by Shallit and Wang's automatic complexity.

		The $p$-value of a model for given data $x$ is the probability that there exists a model with
		as few states, accepting as few words, fitting uniformly randomly selected data $y$.

		Deterministic and nondeterministic automata can give different optimal models.
		For $x=011\, 110\, 110\, 11$, the best deterministic model has $p$-value $0.3$,
		whereas the best nondeterministic model has $p$-value $0.04$.

		In the nondeterministic case, counting paths and counting words can give different optimal models.
		For $x=01100\, 01000$, the best path-counting model has $p$-value $0.79$,
		whereas the best word-counting model has $p$-value $0.60$.
	\end{abstract}



	\section{Introduction}
		Shallit and Wang \cite{MR1897300} introduced \emph{automatic complexity} (defined below) as a computable alternative to Kolmogorov complexity.
		They considered deterministic automata, whereas Hyde and Kjos-Hanssen \cite{Kjos-EJC} studied the nondeterministic case,
		which in some ways behaves better.

		Unfortunately, even nondeterministic automatic complexity is somewhat inadequate.
		The word $00010000$ has maximal nondeterministic complexity among all binary strings of length 8. However, intuitively it is quite simple.
		One way to remedy this situation is to consider a \emph{structure function} analogous to that for Kolmogorov complexity.
		The latter was introduced by Kolmogorov at a 1973 meeting in Tallinn and studied by
		Vereshchagin and Vit\'anyi \cite{MR2103496}, Rissanen \cite{Rissanen}, and Staiger \cite{StaigerTCS}.

		Here we show that some notions in this area, in the non-deterministic setting,
		depend on whether we are counting accepting words or accepting paths.
		This is interesting because counting words is most efficient for compression, whereas counting paths seems to lead to more time-efficient computability.

		Several results are \computer. We do not know of short computer proofs (certificates) in most cases, so we only include the claim that
		the result was \computer.

		\begin{figure*}
			\begin{tikzpicture}[shorten >=1pt,node distance=1.5cm,on grid,auto]
				\node[state,initial, accepting] (q_1)   {$q_1$}; 
				\node[state] (q_2)     [right=of q_1   ] {$q_2$}; 
				\node[state] (q_3)     [right=of q_2   ] {$q_3$}; 
				\node[state] (q_4)     [right=of q_3   ] {$q_4$};
				\node        (q_dots)  [right=of q_4   ] {$\ldots$};
				\node[state] (q_m)     [right=of q_dots] {$q_m$};
				\node[state] (q_{m+1}) [right=of q_m   ] {$q_{m+1}$}; 
				\path[->] 
					(q_1)     edge [bend left]  node           {$x_1$}     (q_2)
					(q_2)     edge [bend left]  node           {$x_2$}     (q_3)
					(q_3)     edge [bend left]  node           {$x_3$}     (q_4)
					(q_4)     edge [bend left]  node [pos=.45] {$x_4$}     (q_dots)
					(q_dots)  edge [bend left]  node [pos=.6]  {$x_{m-1}$} (q_m)
					(q_m)     edge [bend left]  node [pos=.56] {$x_m$}     (q_{m+1})
					(q_{m+1}) edge [loop above] node           {$x_{m+1}$} ()
					(q_{m+1}) edge [bend left]  node [pos=.45] {$x_{m+2}$} (q_m)
					(q_m)     edge [bend left]  node [pos=.4]  {$x_{m+3}$} (q_dots)
					(q_dots)  edge [bend left]  node [pos=.6]  {$x_{n-3}$} (q_4)
					(q_4)     edge [bend left]  node           {$x_{n-2}$} (q_3)
					(q_3)     edge [bend left]  node           {$x_{n-1}$} (q_2)
					(q_2)     edge [bend left]  node           {$x_n$}     (q_1);
			\end{tikzpicture}
			\caption{
				A nondeterministic finite automaton that only accepts one word
				$x= x_1x_2x_3x_4 \ldots x_n$ of length $n = 2m + 1$.
			}
			\label{fig1}
		\end{figure*}
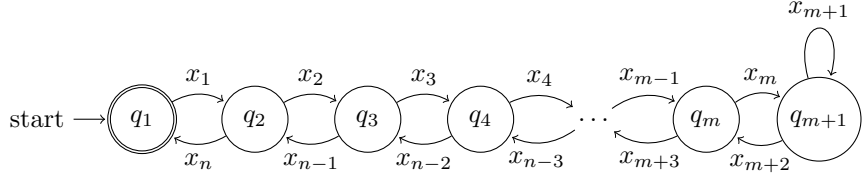

		We let $L(M)$ denote the language recognized by the automaton $M$.
		\begin{df}[Shallit and Wang \cite{MR1897300}]
			The \emph{automatic complexity} of a finite binary word \(x=x_1\dots x_n\) is 
			the least number \(A(x)\) of states of a {deterministic finite automaton} \(M\) such that 
			\[
				L(M)\cap\{0,1\}^n = \{x\},
			\]
			that is, \(x\) is the only word of length \(n\) accepted by \(M\).
			If we do not require the transition function of $M$ to be total, we obtain the \emph{nontotal automatic complexity} $A^-(x)$.
		\end{df}

		We will consider model selection in three distinct \emph{modes}:
		\begin{enumerate}
			\item the deterministic mode $\delta$,
			\item the path-counting nondeterministic mode $\pathc$, and
			\item the word-counting nondeterministic mode $\wordc$.
		\end{enumerate}
		Formally, we could take $\{1,2,3\}=\{\delta, \pathc, \wordc\}$.
		\begin{df}\label{acc}
			The number of acceptances $\Acc_n^\mu(M)$ at length $n$ for an NFA $M$ in mode $\mu$ is defined as follows.
			\begin{itemize}
				\item If $\mu$ is the \emph{deterministic mode} then $\Acc_n^\mu(M)$ is $\infty$ (or undefined) if $M$ is not deterministic.
					If $M$ is deterministic then $\Acc_n^\mu(M)$ is the number of words of length $n$ accepted by $M$.
				\item If $\mu$ is the \emph{path-counting nondeterministic mode} then
					$\Acc_n^\mu(M)$ is the number of paths of length $n$ leading to an accept state of $M$.
				\item If $\mu$ is the \emph{word-counting nondeterministic mode} then
					$\Acc_n^\mu(M)$ is the number of words of length $n$ accepted by $M$.
			\end{itemize}
		\end{df}
		Following Kolmogorov, we shall rarely consider more fine-grained acceptance counting than just by powers of $b$.
		So we define
		\(
			\logAcc_n^\mu(M, b) =
			\ceil{
				\log_b
				\Acc_n^\mu(M)
			}
		\).
		\begin{df}[{\cite{Kjos-EJC}}]\label{precise}
			The path-counting nondeterministic automatic complexity $A^{\pathc}_N(w)$ of a word $w$ is the minimum number of states of an NFA $M$
			such that $M$ accepts $w$ and
			$\Acc_{\abs{w}}^\pathc(M)=1$.
		\end{df}

		We assume that our NFAs are not generalized, i.e., they have no $\eps$-transitions.
		\begin{thm}
			It does not matter for $A_N^{\pathc}$ or $A_N^{\wordc}$ whether $\eps$-transitions are allowed.
		\end{thm}
		\begin{proof}
			Given an automaton $M$ using $\eps$-transitions, we define another automaton $M'$ not using any $\eps$-transitions.
			We put a transition in $M'$ between states $q_1$ and $q_2$ labeled $i$ if there is some path from $q_1$ to $q_2$ in $M$ whose labels concatenate to $i$
			under the obvious rule that $i\eps=i=\eps i$.
		\end{proof}
		We assume our automata have only a single accept state.

		The definition of $A^{\pathc}_N$ is not robust under permutation of quantifiers, in the following sense.
		\begin{df}
			Let $A^\dag(x)$ be the minimum number of states of an NFA such that
			$x$ is the only string of length $n$ that is accepted along exactly one path (but other strings may be accepted among more than one path).
		\end{df}
		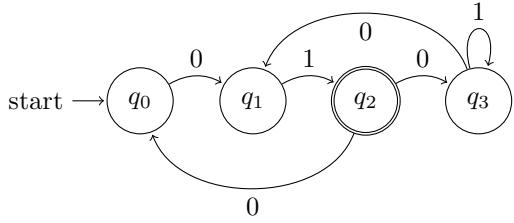
\begin{figure}
			\centering
			\begin{tikzpicture}[shorten >=1pt,node distance=1.5cm,on grid,auto]
				\node[state, initial] (q_0) {$q_0$};
				\node[state] (q_1) [right=of q_0] {$q_1$};
				\node[state, accepting] (q_2) [right=of q_1] {$q_2$};
				\node[state] (q_3) [right=of q_2] {$q_3$};
				\path[->]
					(q_0) edge [bend left]  node {$0$} (q_1)
					(q_1) edge [bend left]  node {$1$} (q_2)
					(q_2) edge [bend left]  node {$0$} (q_3)
					(q_3) edge [loop above] node {$1$} (q_3)
					(q_2) edge [bend left=70]  node {$0$} (q_0)
					(q_3) edge [bend right=70]  node {$0$} (q_1);
			\end{tikzpicture}
			\caption{
				An automaton accepting only $x=010111010$ along exactly one path, but accepting other words along multiple paths, giving Theorem \ref{daggerThm}.
			}
			\label{dagger}
		\end{figure}
		When considering automatic complexity it is often sufficient to replace an automaton by a \emph{state sequence}. A state sequence is a sequence of states, typically the sequence of states visited by the automaton during the processing a word. For computational purposes we may represent a state sequence $q_0,\dots,q_n$ as a sequence of nonnegative integers $s=s_0\dots s_n$ with the property that $s_i\le\max_{j<i}s_j+1$.

		\begin{thm}[\computer]\label{aaaa}
			$A_N^{\pathc}(010111010)=5$.
		\end{thm}
		\begin{thm}\label{daggerThm}
			There is an $x$ such that $A^\dag(x)<A_N^{\pathc}(x)$.
		\end{thm}
		\begin{proof}
			Consider $x=010111010$.
			By Theorem \ref{aaaa}, $A_N^{\pathc}(x)=5$.
			As the state sequence 0123333120 witnesses, $A^\dag(x)\le 4$. (See Figure \ref{dagger}.)
		\end{proof}
		\begin{df}\label{df:KayleighGraph}
			Let $n = 2m + 1$ be a positive odd number, $m\ge 0$.
			A finite automaton of the form given in Figure \ref{fig1} for some choice of symbols $x_1,\dots,x_n$ and states
			$q_1,\dots,q_{m+1}$
			is called a \emph{Kayleigh graph}.
		\end{df}
		\begin{thm}[Hyde \cite{Kjos-EJC}]\label{Hyde}
			The nondeterministic automatic complexity $A^{\pathc}_N(x)$ of a word $x$ of length $n$ satisfies
			\[
				A^{\pathc}_N(x) \le {\lfloor} n/2 {\rfloor} + 1\text{.}
			\]
		\end{thm}
		\begin{proof}
			If $n$ is odd, then a Kayleigh graph witnesses this inequality.
			If $n$ is even, a slight modification suffices.
		\end{proof}
		\begin{df}
			Let $\mu\in\{\wordc,\pathc\}$.
			Suppose $x$ is a binary word of length $n$.
			$S^{\mu}_x$ is defined to be the set of pairs of integers $(q,m)$ such that
			there exists an NFA $M$ with $x\in L(M)$, at most $q$ states, and $\logAcc_n^\mu(M, 2)\le m$.
		\end{df}
		We note that $S^{\mu}_x$ has the upward closure property
		\[
			q\le q', m\le m', (q,m)\in S^{\mu}_x \quad\Longrightarrow\quad (q',m')\in S^{\mu}_x.
		\]
		From $S^{\mu}_x$ we can define the structure function $h^{\mu}_x$ and
		the dual structure function $h_x^{*\mu}$.
		The definition was presented to us by Vereshchagin (personal communication, 2014), inspired by \cite{MR2103496}.
		\begin{df}
			\begin{eqnarray*}
				h^{*\mu}_x(m) &=& \min\{k : (k,m)\in S^{\mu}_x \},\\
				h^{\mu}_x(k)  &=& \min\{m: (k,m)\in S^{\mu}_x\}.
			\end{eqnarray*}
		\end{df}
		Note that $S^{\pathc}_x \subseteq S^{\wordc}_x$ and hence 
		$h^{\wordc}_x(k)\le h^{\pathc}_x(k)$ and
		$h^{*\wordc}_x(m)\le h^{*\pathc}_x(m)$ for each $k$, $m$, and $x$,.
		Upper bounds on $h^{*\wordc}(m)$, generalizing the $m = 0$ case covered by Hyde's Theorem \ref{Hyde}, were studied in \cite{Kjos-TCS}.
		It may be observed that the proofs given there are based on counting accepting paths and hence apply equally to $h^{*\pathc}_x$.

		\begin{df}
			For a word $x$, the word-based nondeterministic automatic complexity of $x$ is defined by
			\[
				A^{\wordc}_N(x) = h_x^{*\wordc}(0).
			\]
		\end{df}
		Equivalently,
		\[
		 	A^{\wordc}_N(x) = A^{\wordc}_N(\{x\}).
		\]
		\begin{con}\label{mainConj}
			There is an $x$ such that $A^{\wordc}_N(x) \ne A^{\pathc}_N(x)$.
		\end{con}
		Conjecture \ref{mainConj} lies at the crossroads of, and indeed was the inspiration for, the following results.
		\begin{itemize}
			\item Theorem \ref{smallThm}: $h^{*\wordc}_x(m)\ne h^{*\pathc}_x(m)$ for the word $000010000$, at $m=1$.
			Conjecture \ref{mainConj} states that there is even such an example with $m=0$.
			\item Theorem \ref{chokoUndKeks}: $A_N^{\wordc}(\F)\ne A_N^{\pathc}(\F)$ for the doubleton $\F=\{0110,1111\}$.
			Conjecture \ref{mainConj} states that there is a singleton example.
		\end{itemize}
		
		\begin{thm}[\computer]\label{tenComp}
			There is no binary word $x$ with $\abs{x}\le 10$ and $A^{\wordc}_N(x) \ne A^{\pathc}_N(x)$.
		\end{thm}

	\section{Structure functions}\label{Igusa}
		We now show that the automatic complexity structure function of a word $x$
		sometimes depends on whether we are counting accepting paths or accepted words.
		\begin{thm}\label{pacific}
			For any word $x=x_1\dots x_n$,
			\[
				h_x^{*\wordc}(n-2)\le 2.
			\]
		\end{thm}
		\begin{proof}
			It suffices to consider the following NFA:
			\[
			\centering
			\begin{tikzpicture}[shorten >=1pt,node distance=1.5cm,on grid,auto]
				\node[state, initial, accepting] (q_0) {$q_0$};
				\node[state] (q_1) [right=of q_0] {$q_1$};
				\path[->]
					(q_0) edge [bend left]  node {$x_1$} (q_1)
					(q_1) edge [loop above]  node {$0$} (q_1)
					(q_1) edge [loop below]  node {$1$} (q_1)
					(q_1) edge [bend left]  node {$x_n$} (q_0);
			\end{tikzpicture}\qedhere
			\]
		\end{proof}
		\begin{thm}[\computer]\label{bbbb}
			Let $x=001011$. Then
			$h_x^{*\pathc}(\abs{x} - 2) = 3$.
		\end{thm}

		\begin{thm}\label{mainThm}
			There is a binary word $x$ of length 6 and an $m$ such that
			$h^{*\wordc}_x(m)< h^{*\pathc}_x(m)$.
		\end{thm}
		\begin{proof}
			Let $x=001011$.
			By Theorem \ref{pacific} and
			Theorem \ref{bbbb},
			\[
				h_x^{*\wordc}(\abs{x}-2)\le 2 < 3 = h_x^{*\pathc}(\abs{x} - 2).\qedhere
			\]
		\end{proof}
		Theorem \ref{mainThm} is optimal for $n$, but not for $m$, as
		Theorem \ref{smallThm} indicates.
		\begin{thm}[\computer]\label{cccc}
			Let $x=000010000$. Then
			$h^{*\pathc}_x(1)=5$.
		\end{thm}
		\begin{thm}\label{smallThm}
			There is a word $x$ such that
			\[
				h^{*\wordc}_x(1)< h^{*\pathc}_x(1).
			\]
		\end{thm}
		\begin{proof}
			Let
			$x=000010000$.
			By Figure \ref{bananaMuffin}, $h^{*\pathc}_x(1)\le 4$. Hence by Theorem \ref{cccc},
			\[
				h^{*\pathc}_x(1)\le 4<5=h^{*\wordc}_x(1).\qedhere
			\]
		\end{proof}
		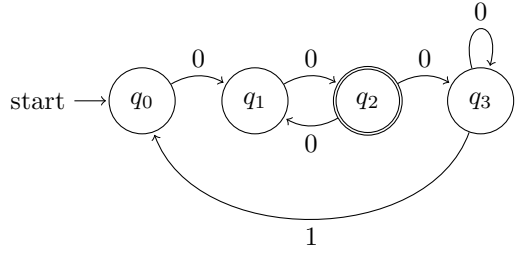
\begin{figure}
			\centering
			\begin{tikzpicture}[shorten >=1pt,node distance=1.5cm,on grid,auto]
				\node[state, initial] (q_0) {$q_0$};
				\node[state] (q_1) [right=of q_0] {$q_1$};
				\node[state, accepting] (q_2) [right=of q_1] {$q_2$};
				\node[state] (q_3) [right=of q_2] {$q_3$};
				\path[->]
					(q_0) edge [bend left]  node {$0$} (q_1)
					(q_1) edge [bend left]  node {$0$} (q_2)
					(q_3) edge [loop above] node {$0$} (q_3)
					(q_2) edge [bend left]  node {$0$} (q_3)
					(q_2) edge [bend left]  node {$0$} (q_1)
					(q_3) edge [bend left=70]  node {$1$} (q_0);
			\end{tikzpicture}
			\caption{
				An automaton accepting $x=0^410^4$ along one path and $0^610^2$ along two paths, giving Theorem \ref{smallThm}.
			}
			\label{bananaMuffin}
		\end{figure}

	\section{Model selection}
		The automatic structure functions are intended to provide statistical explanations for words. The best explanation for a word $x$ is the automaton witnessing a value of the structure function that is unusually low, compared to other words $y$. It turns out that the phenomenon of Theorem \ref{mainThm} also applies to such best explanations.

		As envisioned by Kolmogorov, structure functions have potential applications in computational statistics.
		We now describe concrete results of our foray into model selection with structure functions for automatic complexity.
		\begin{df}
			Let $X$ be a uniform random variable on $[b]^n$.
			The $b$-ary $p$-value achieved by an NFA $M$ at a length $n$ in mode $\mu$ is the probability that
			$X$ is accepted by some NFA $N$ such that
			$N$ has no more states than $M$, and $\logAcc^\mu_n(N, b)\le \logAcc^\mu_n(M, b)$.
		\end{df}
		\begin{df}
			An NFA $M$ is an \emph{optimal $b$-ary model for $x$ in mode $\mu$} if
			$M$ accepts $x$ and
			$M$ achieves the minimal $b$-ary $p$-value at length $\abs{x}$ in mode $\mu$ among all NFAs that accept $x$.
			The \emph{$b$-ary explanation of $x$ in mode $\mu$} is the set of all optimal $b$-ary models for $x$ in mode $\mu$.
		\end{df}
		Often, we take $b$ to be the least integer such that $x$ is a word in the alphabet $[b]$.
		For binary words, we usually take $b=2$, even in the case of the word $0^n$.

		\begin{thm}[\computer]\label{eeee}
			Let $x=01111011011$.
			In both the path-counting mode and the deterministic mode the optimal number of states for $x$ is 3.
			The only optimal state sequence for $x$ in the path-counting mode is 012120120120, giving $m=2$ and $p$-value $0.04$.
			The only optimal state sequence for $x$ in the deterministic mode is 012020120120, giving $m=4$ and $p$-value $0.30$.
		\end{thm}
		Theorem \ref{eeee} immediately gives an interesting corollary.
		\begin{thm}\label{stats}
			There is an $x$ such that
			the explanation of $x$ in deterministic mode and
			the explanation of $x$ in path-counting mode
			are disjoint.
		\end{thm}
		See Figure \ref{optimal-nondet} for illustration of Theorems \ref{eeee} and \ref{stats}.
		\begin{figure}
			\centering
			\begin{tikzpicture}[shorten >=1pt,node distance=1.5cm,on grid,auto]
				\node[state, initial, accepting] (q_0)                {$q_0$}; 
				\node[state]                     (q_1) [right=of q_0] {$q_1$}; 
				\node[state]                     (q_2) [right=of q_1] {$q_2$}; 
				\node (q_3) [right=of q_2]{};
				\path[->] 
					(q_0) edge                node {0} (q_1)
					(q_1) edge [bend left]    node {1} (q_2)
					(q_2) edge [bend left]    node {1} (q_1)
					(q_2) edge [bend left=70] node {1} (q_0);
				\node[state, initial, accepting] (r_0) [right=of q_3] {$r_0$}; 
				\node[state]                     (r_1) [right=of r_0] {$r_1$}; 
				\node[state]                     (r_2) [right=of r_1] {$r_2$}; 
				\path[->] 
					(r_0) edge                node {0} (r_1)
					(r_1) edge                node {1} (r_2)
					(r_2) edge [bend left=50] node {1} (r_0)
					(r_0) edge [bend left=50] node {1} (r_2);
			\end{tikzpicture}
			\caption{
				Optimal models for 01111011011 in the path-counting (left) and deterministic (right) modes (Theorem \ref{stats}).
			}
			\label{optimal-nondet}\label{optimal-det}
		\end{figure}
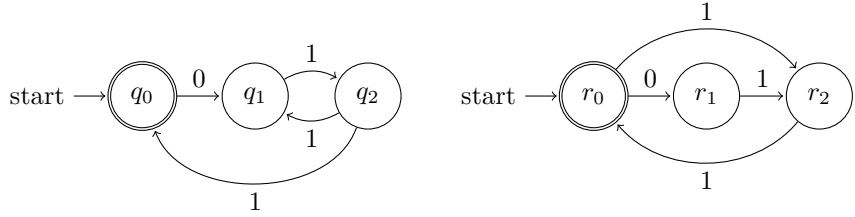
		\begin{thm}[\computer]\label{ffff}
			The optimal number of states for $x=0110001000$ in the path-counting mode is 4, corresponding to $m=2$ and a $p$-value of 0.79.
			The optimal number of states for $x=0110001000$ in the word-counting mode is 2, corresponding to $m=7$ and a $p$-value of 0.6.
		\end{thm}
		The corollary we seek is now immediate from Theorem \ref{ffff}.
		\begin{thm}\label{word-stats}
			There is an $x$ such that
			the explanation of $x$ in word-counting mode and
			the explanation of $x$ in path-counting mode
			are disjoint.
		\end{thm}
		 See Figure \ref{optimal-word} for an illustration of Theorems \ref{ffff} and \ref{word-stats}.

		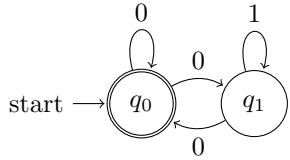
\begin{figure}
			\centering
			\begin{tikzpicture}[shorten >=1pt,node distance=1.5cm,on grid,auto]
				\node[state, initial, accepting] (q_0)                {$q_0$};
				\node[state]                     (q_1) [right=of q_0] {$q_1$}; 
				\path[->] 
					(q_0)     edge [loop above] node           {0} ()
					(q_0)     edge [bend left]  node           {0} (q_1)
					(q_1)     edge [bend left]  node           {0} (q_0)
					(q_1)     edge [loop above] node           {1} ();
			\end{tikzpicture}
			\caption{
				An optimal model for $x=0110001000$ in the word-counting mode (Theorem \ref{word-stats}).
				Note the use of multiple paths.
			}
			\label{optimal-word}
		\end{figure}

	\section{Determinism and automatic complexity}
		In \cite{Kjos-SIAM} we give an example of a word $x$ such that $A^-(x)-A_N^{\pathc}(x)=2$.
		We conjecture that the differences $A^-(x)-A_N^{\pathc}(y)$
		are unbounded as $\abs{y}\to\infty$.
		However, for most words, the difference between $A^-$ and $A_N^{\pathc}$ is small.
		Let $[b]=\{0,\dots,b-1\}$ and let $\abs{X}$ denote the cardinality of a set $X$.
		We show that most words have $A^-$-complexity at most $(\frac12+\eps)n$ in the following sense.
		\begin{thm}
			For each $\eps>0$ and integer $b\ge 1$,
			\[
				\lim_{n\to\infty}
					b^{-n}
					\left|\left\{x\in [b]^n : \frac{A^-(x)}n \le \frac12+\frac1{2b}+\eps\right\}\right|
				= 1.
			\]
		\end{thm}
		\begin{proof}[Proof sketch.]
			The idea is \emph{derandomization}, or perhaps more accurately \emph{determinization}, of Kayleigh graphs (Figure \ref{fig1}).
			Whenever there is a state with nondeterministic out-behavior, split
			it into two states as in Figure \ref{fig1explode}. This will only happen about a fraction $\frac1b$ of the time,
			so the total number of states will be about
			\[
				\frac{n}2 + \frac{n}2\cdot\frac1b = \left(\frac12+\frac1{2b}\right)n. 
			\]
			By the Law of Large Numbers, the statement of the Theorem follows.
		\end{proof}

		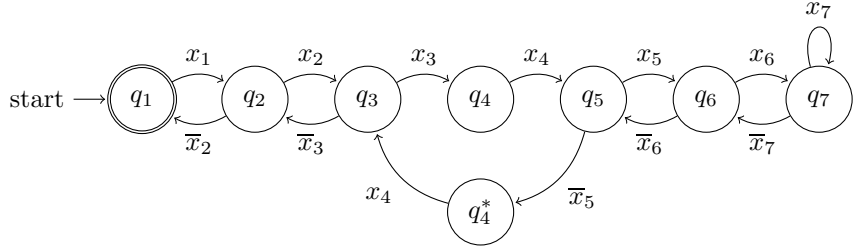
\begin{figure}
			\begin{tikzpicture}[shorten >=1pt,node distance=1.5cm,on grid,auto]
				\node[state,initial, accepting] (q_1)   {$q_1$}; 
				\node[state] (q_2)     [right=of q_1   ] {$q_2$}; 
				\node[state] (q_3)     [right=of q_2   ] {$q_3$}; 
				\node[state] (q_41)     [right=of q_3   ] {$q_{4}$};
				\node[state] (q_42)    [below=of q_41   ] {$q_{4}^*$};
				\node[state] (q_5)     [right=of q_4   ] {$q_5$};
				\node[state] (q_6)     [right=of q_5   ] {$q_6$};
				\node[state] (q_7) [right=of q_6   ] {$q_7$}; 
				\path[->] 
					(q_1)     edge [bend left]  node           {$x_1$}     (q_2)
					(q_2)     edge [bend left]  node           {$x_2$}     (q_3)
					(q_3)     edge [bend left]  node           {$x_3$}     (q_41)
					(q_41)    edge [bend left]  node  {$x_4$}     (q_5)
					(q_5)     edge [bend left]  node   {$x_{5}$} (q_6)
					(q_6)     edge [bend left]  node  {$x_6$}     (q_7)
					(q_7)     edge [loop above] node           {$x_{7}$} ()
					(q_7)     edge [bend left]  node  {$\overline x_{7}$} (q_6)
					(q_6)     edge [bend left]  node  {$\overline x_{6}$} (q_5)
					(q_5)     edge [bend left]  node  {$\overline x_{5}$} (q_42)
					(q_42)    edge [bend left]  node           {$x_{4}$} (q_3)
					(q_3)     edge [bend left]  node           {$\overline x_{3}$} (q_2)
					(q_2)     edge [bend left]  node           {$\overline x_{2}$}     (q_1);
			\end{tikzpicture}
			\caption{
				A deterministic finite automaton that only accepts one word
				$x= x_1x_2x_3x_4x_5x_6x_7\overline x_7\overline x_6\overline x_5 x_4\overline x_3\overline x_2$ of length $n = 13$.
				It is obtained by ``exploding'' the state $q_4$ in a Kayleigh graph (Figure \ref{fig1}).
			}
			\label{fig1explode}
		\end{figure}

		We also know that $A^{\pathc}_N$ and $A^{\wordc}_N$ have the same sharp upper bound. The argument in \cite{Kjos-EJC},
		to the effect that $n/2+1$ is sharp,
		applies to them equally.
	\section{Automatic complexity of doubletons}
		\begin{df}
			The word-based automatic complexity $A^{\wordc}_N(\F)$ of
			a finite set $\F\subseteq\{0,1\}^n$ to be the minimum number of states of an NFA $M$ such that
			\[
				L(M)\cap\{0,1\}^n=\F.
			\]
			The path-based automatic complexity $A^{\pathc}_N(\F)$ is the minimum number of states of an NFA $M$ such that
			in addition $M$ has only $|\F|$ many accepting paths of length $n$.
		\end{df}

		This generalizes automatic complexity from the case where $\F$ is a singleton.
		Clearly $A^{\wordc}_N \le A^{\pathc}_N$. We shall see in Theorem \ref{chokoUndKeks} that
		$A^{\wordc}_N(\F) \ne A^{\pathc}_N(\F)$ in general when $|\F|=2$.
		We conjectured in Conjecture \ref{mainConj} that
		$A^{\wordc}_N(\F) \ne A^{\pathc}_N(\F)$ for some $\F$ with $|\F|=1$.

		\begin{figure}
			\centering
			\begin{tikzpicture}[->,>=stealth',shorten >=1pt,auto,node distance=2.8cm,
			                    semithick]
			  \tikzstyle{every state}=[draw=black,text=black]

			  \node[state,accepting] (A)              {};
			  \node[state]           (B) [above of=A] {};
			  \node[initial,state]   (D) [below of=A] {};
			  \node[state]           (C) [right of=A] {};
			  \node[state]           (E) [above of=C] {};

			  \path (A) edge [bend left]  node {$x_2$} (B)
			        (B) edge [loop above] node {$x_3$} (B)
			            edge [bend left]  node {$x_4$} (A)
			        (C) edge              node {$y_2$} (E)
			        (D) edge              node {$x_1$} (A)
			            edge              node {$y_1$} (C)
			        (E) edge [loop above] node {$y_3$} (E)
			            edge              node {$y_4$} (A);
			\end{tikzpicture}
			\caption{An automaton witnessing the Chambers--Hyde bound (Theorem \ref{thermalbad}) for $n=4$ and $f=2$.}\label{chambers--hyde}
		\end{figure}
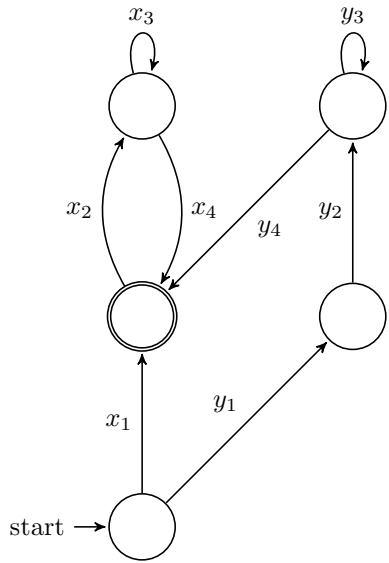
		\begin{thm}[Hyde \cite{Kjos-EJC} ($f=1$), Chambers \cite{Chambers} ($f=2$)]\label{thermalbad}
			The automatic complexity of a set $\F\subseteq \{0,1\}^n$
			(with one accept state allowed) of size $f$ satisfies
			\[
				A^{\pathc}_N(\F)\le f\lfloor n/2\rfloor + 1.
			\]
		\end{thm}
		\begin{proof}[Proof sketch.]
			The proof is easy upon consideration of Figures \ref{fig1} and \ref{chambers--hyde}.
		\end{proof}

		\begin{figure}
			\centering
			\begin{tikzpicture}[shorten >=1pt,node distance=1.5cm,on grid,auto]
				\node[state, initial] (q_1)   {$q_0$}; 
				\node[state] (q_2)     [right=of q_1   ] {$q_1$}; 
				\node[state, accepting] (q_3)     [right=of q_2   ] {$q_2$}; 
				\path[->] 
					(q_1)     edge [bend left]  node           {1}     (q_2)
					(q_2)     edge [bend left]  node           {1}     (q_3)
					(q_3)     edge [bend left]  node           {1}     (q_2)
					(q_2)     edge [bend left]  node           {1}     (q_1)
					(q_1)     edge [bend left=70]  node           {0}     (q_3);
			\end{tikzpicture}
			\caption{
				The optimal automaton $C$ with $L(C)\cap\{0,1\}^4=\{0110,1111\}$ by necessity accepts 1111 along two paths,
				giving Theorem \ref{chokoUndKeks}.
			}
			\label{hot}
		\end{figure}
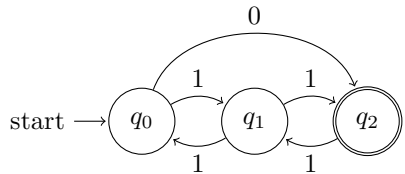
		
		\begin{thm}[\computer]\label{dddd}
			Let $\F=\{0110, 1111\}$. For any NFA $M$ such that $L(M)\cap \{0,1\}^4=\F$, $M$ has at least 3 states, and
			if $M$ has 3 states then $M$ has at least 3 accepting paths of length 4.
		\end{thm}
		\begin{thm}\label{chokoUndKeks}
			There is an $n$ and a finite set $\F\subseteq\{0,1\}^n$ such that
			\[
				A_N^{\wordc}(\F)\ne A_N^{\pathc}(\F).
			\]
		\end{thm}
		\begin{proof}
			Let $n=4$ and let $\F=\{0110, 1111\}$. By Theorem \ref{dddd}, the automaton $C$ in Figure \ref{hot} is optimal.
			Thus
			\[
				A_N^{\wordc}(\F) = 3 < 4 \le A_N^{\pathc}(\F).\qedhere
			\]
		\end{proof}

	\section*{Acknowledgments}
		We thank
		Greg Igusa for contributing ideas to
		Section \ref{Igusa}.
		This work
		was partially supported by
		a grant from the Simons Foundation (\#315188 to Bj\o rn Kjos-Hanssen).
		This material is based upon work supported by the National Science Foundation under Grant No.\ 1545707.
	\bibliographystyle{plain}
	\bibliography{few-paths-fewer-words}
\end{document}